\documentclass[10pt]{article}
\pdfoutput=1
\usepackage{amsmath, amsthm, amssymb, bbm, setspace, mathrsfs, color, listings, graphicx, natbib, multicol, dcolumn}
\usepackage[font=footnotesize,labelfont=bf]{caption}
\usepackage{float,subcaption}
\usepackage[page,title]{appendix}
\usepackage[left=0.7in,right=1.0in,top=0.7in,bottom=1.0in,nohead]{geometry}

\newif\iffinalversion
\finalversiontrue
\iffinalversion
	\usepackage[colorlinks=true]{hyperref} 
\else
	\usepackage[colorlinks=false]{hyperref} 
\fi

\newtheorem{thm}{Theorem}
\newtheorem{asm}{Assumption}
\newtheorem{lem}[thm]{Lemma}

\newtheorem{prop}[thm]{Proposition}

\newcommand{\utwi}[1]{\mbox{\boldmath $ #1$}}
\newcommand{\Y}{{\utwi{Y}}}
\newcommand{\X}{{\utwi{X}}}
\providecommand{\e}[1]{\ensuremath{\times 10^{#1}}}

\newcolumntype{d}{D{.}{.}{9}} 

\begin{document}

\title{Change Points via Probabilistically Pruned Objectives}
\author{Nicholas A. James and David S. Matteson \\ Cornell University\footnote{
James is a PhD Candidate,
School of Operations Research and Information Engineering,
Cornell University,
206 Rhodes Hall,
Ithaca, NY 14853
(Email: \href{mailto:nj89@cornell.edu}{nj89@cornell.edu}; Web: \url{https://courses.cit.cornell.edu/nj89/}).
Matteson is an Assistant Professor,
Department of Statistical Science,
Cornell University,
1196 Comstock Hall,
Ithaca, NY 14853
(Email: \href{mailto:matteson@cornell.edu}{matteson@cornell.edu}; Web: \url{http://www.stat.cornell.edu/\~matteson/}).}}
\date{}
\maketitle
\begin{abstract}
	The concept of homogeneity plays a critical role in statistics, both in its applications as well as its theory. Change point analysis is a statistical 
	tool that aims to attain homogeneity within time series data. This is accomplished through partitioning the time series into a number of contiguous 
	homogeneous segments. The applications of such techniques range from identifying chromosome alterations to solar flare detection. In this manuscript we present 
	a general purpose search algorithm called cp3o that can be used to identify change points in multivariate time series. This new search procedure can be 
	applied with a large class of goodness of fit measures. Additionally, a reduction in the computational time needed to identify change points is accomplish 
	by means of probabilistic pruning. With mild assumptions about the goodness of fit measure this new search algorithm is shown to generate consistent estimates 
	for both the number of change points and their locations, even when the number of change points increases with the time series length.\par
	
	A change point algorithm that incorporates the cp3o search algorithm and E-Statistics, e-cp3o, is also presented. The only distributional assumption that the 
	e-cp3o procedure makes is that the absolute $\alpha$th moment exists, for some $\alpha\in(0,2)$. Due to this mild restriction, the e-cp3o procedure can be applied 
	to a majority of change point problems. Furthermore, even with such a mild restriction, the e-cp3o procedure has the ability to detect \textit{any} type of 
	distributional change within a time series. Simulation studies are used to compare the e-cp3o procedure to other parametric and nonparametric change point procedures, 
	we we highlight applications of e-cp3o to climate and financial datasets.
\end{abstract}

\par\vfill\noindent

\clearpage\pagebreak\newpage

\section{Introduction}\label{intro}
The analysis of time ordered data, referred to as time series, has become a common practice in both academic and industrial settings. The applications of such 
analysis span may different fields, each with its own analytical tools. Such fields  include network security \citep{blazek01,siris06}, fraud detection \citep{akoglu10,fawcett97}, 
financial modeling \citep{andreou02,dias04}, climate analysis \citep{wang13}, astronomical observation \citep{friedman96,xie13}, and many others.\par

However, when analysis is performed it is generally assumed that the data adheres to some form of homogeneity. This could mean a range of things, depending upon the 
application area. Some common types of assumed homogeneity include: constant mean, constant variance, and strong or weak stationarity. Depending on the nature of 
these assumptions it may not be appropriate, or practical, to apply a given analytical procedure to many different types of time series. For instance, an algorithm 
that assumes weak stationarity would not be suitable for analyzing data that follows a Cauchy distribution, because of its infinite expectation. Furthermore many 
time series of real data can be seen, even through visual inspection, to violate such homogeneity conditions.\par

Results obtained under such model misspecification can vary in their degree of inaccuracy \citep{chong95}. The resulting bias from such misspecification is one of the reasons 
for the current resurgence of change point analysis. Change point analysis attempts to partition a time series into homogeneous segments. Once again the definition 
of homogeneity will depend upon the application area. In this paper we will use a notion of homogeneity that is common in the statistical literature. We will say 
that a segment is homogeneous if all of its observations are identically distributed. Using this definition of homogeneity, change point analysis can be performed 
in a variety of ways.\par

In this paper we consider the following formulation of the offline multiple change point problem. Let $Z_1, Z_2,\dots, Z_T\in\mathbb R^d$ be a length $T$ sequence of independent 
$d$-dimensional time ordered observations. The dimension of our observations is arbitrary, but assumed to be fixed. Additionally, let $F_0, F_1, F_2, \dots,$ be a (possibly infinite) 
sequence of distributional functions, such that $F_i\neq F_{i+1}.$ It is also assumed that in the sequence of observations, 
there is at least one distributional change. Thus, there exists $k(T)\ge1$ time indices $0=\tau_{0,T}<\tau_{1,T}<\dots<\tau_{k(T),T}<\tau_{k(T)+1,T}=T$, such that 
$Z_i\stackrel{iid}{\sim} F_j$, for $\tau_{j,T}<i\le\tau_{j+1,T}$. From this notation it is clear 
that the locations of change points $\tau_{j,T}$ depend upon the sample size. However, we will usually suppress this dependence and use the notation $\tau_j$ for simplicity. 
The challenge of multiple change point analysis is to provide a good estimate of both the number of change points, $k(T)$, as well as their respective locations, 
$\tau_1, \tau_2,\dots \tau_{k(T)}$. In some cases it is also necessary to provide some information about the distributions $F_0,\dots, F_{k(T)}$. However, once a segmentation is 
provided it is usually straight-forward to obtain such information.\par

A popular approach is to fit the observed data to a parametric model. In this setting a change point corresponds to a change in the monitored parameter(s). Earlier 
work in this area assumes Gaussian observations and proceeds to partition the data through the use of maximum likelihood \citep{maboudou13}. More recently, extensions to other 
members of the Exponential family of distributions and beyond have been considered \citep{chen11}. In general, all of these approaches rely on the existence of a likelihood 
function with an analytic expression. Once the likelihood function is known, analysis is reduced to finding a computationally efficient way to maximize the likelihood over a set 
of candidate parameter values.\par

Parametric approaches however, rely heavily upon the assumption that the data behaves according to the predefined model. If this is not the case, then the degree of 
bias in the obtained results is usually unknown \citep{pitarakis04}. In practice, it is almost always difficult, if not impossible, to test for adherence to these assumptions. 
Under such settings, performing nonparametric analysis is a natural way to proceed \citep{brodsky93}. Since nonparametric approaches make much weaker assumptions than their 
parametric counterparts they can be used in a much wider variety of settings; for example, the analysis of internet traffic data, where there is no commonly accepted 
distributional model. Even though these methods do not directly impose a distributional model for the data, they do make their own types of assumptions \citep{zou14}. For 
instance, a common assumption is the existence of a density function, which then allows practitioners to perform maximum likelihood estimation by using estimated densities. 
However, estimation becomes inaccurate and time consuming when the dimension of the time series increases \citep{hastie09}.\par

Performing multiple change point analysis can easily become computationally intractable. Usually the number of true change points is not known beforehand. However, even if 
such information were provided, finding the locations is not a simple task. For instance, if it is known that the time series contains $k$ change points then there are 
$\mathcal O(T^k)$ possible segmentations. Thus naive approaches to find the best segmentation quickly become impractical. More refined techniques must therefore be employed 
in order to obtain change point estimates in a reasonable amount of time.\par

Most existing procedures for performing retrospective multiple change point analysis can be classified as belonging to one of two groups. The first consists of search 
procedures which will return what are referred to as \emph{approximate} solutions, while the second consist of those that produce \emph{exact} solutions. As indicated by the 
name, the approximate procedures tend to produce suboptimal segmentations of the given time series. However, their benefit is that they tend to have provably much lower 
computational complexity than procedures that return optimal segmentations.\par

Approximate search algorithms tend to rely heavily on a subroutine for finding a single change point. Estimates for multiple change point locations are produced by iteratively 
applying this subroutine. Such algorithms are commonly referred to as binary segmentation algorithms. In many cases it can be shown that binary segmentation algorithms have a 
complexity of $\mathcal O(T\log T)$. This type of approach to multiple change point analysis was introduced by \cite{vostrikova81} and has since been adapted by many others. 
Such adaptations include the Circular Binary Segmentation approach of \cite{olshen04} as well as the E-Divisive approach of \cite{matteson13}. The Wild Binary Segmentation 
approach of \cite{fryzlewicz13} is a variation of binary segmentation that utilizes random intervals in an attempt to further reduce computational time. An extension of this 
approach to multivariate multiplicative time series called Sparsified Binary Segmentation has been produced by \cite{cho12}. Each of these procedures have been shown to produce 
consistent estimates of both the number and locations of change points under a variety of model conditions.\par

Exact search algorithms return segmentations that are optimal with respect to a prespecified goodness of fit measure, such as a log likelihood function. The naive approach of 
searching over all possible segmentations quickly becomes impractical for relatively small time series with a few change points. Therefore, in order to achieve a reasonable 
computational cost, the utilized goodness of fit measures often satisfy Bellman's Principle of Optimality \citep{bellman52}, and can thus be optimized through the use of dynamic 
programming. However, in most cases the ability to obtain this optimal solution comes with a computational cost. Usually this results in at least $\mathcal O(T^2)$ computational 
complexity. Examples of exact algorithms include the Kernel Change Point algorithm, \citep{cappe07} and \citep{arlot12}, and the MultiRank algorithm \citep{fong11}. The complexity 
of these algorithms also depends upon the number of identified change points. However, a method introduced by \cite{jackson05}, as well as the PELT algorithm of \cite{killick12} 
can both obtain optimal segmentations with running times that are independent of the number of change points. An additional benefit of the PELT approach is that under certain 
conditions it is shown to have an expected running time that is linear in the length of the time series.\par

The second aspect of multiple change point analysis is the determination of the number of change points. The first technique that is commonly used by 
approximate search algorithms is hypothesis testing. This method continues to identify change points until they are unable to reject the null hypothesis of no change. Such an 
approach however, is not well suited for many procedures that use exact search algorithms, since many identify all change points at once, 
instead of sequentially, as is the case with binary segmentation. Many change point algorithms that use an exact search procedure instead turn to penalized optimization. The 
reasoning behind penalization is that a more complex model, in this case one with more change points, will better fit the observed data. The the penalty thus helps to guard 
against over segmentation. \cite{yao87} showed that using the Schwarz Criterion can produce a consistent estimate of the number of change points. It has since become popular 
to maximize a penalized goodness of fit measure of the following form,
\begin{equation}\label{penGof}
    S(k,\beta) = \max_{\tau_1<\tau_2<\dots<\tau_k}\ \sum_{i=1}^k\mathcal C(C_i) + {\mathcal P}(k),
\end{equation}
for a penalty function ${\mathcal P}(\cdot)$ and measure of segmentation quality $\mathcal C(\cdot)$. A common choice for the penalty function is ${\mathcal P}(k)=-\beta k$, for 
some user defined positive constant $\beta.$ This type of penalization only takes into consideration the 
number of change points, and not their location. There are other penalization approaches that not only consider the number of change points, but also the change point locations. 
See for instance \cite{zhang07} and \cite{hannart12}.\par

An alternative to penalization is to instead generate all optimal segmentations with $k$ change points, up to some prespecified upper limit. This corresponds to evaluating 
$S(k,0)$ from Equation \ref{penGof} for a range of $k$ values. However, depending on the choice of $\mathcal C(\cdot)$ it may not be possible to efficiently calculate $S(k,0)$ 
for a range of $k$ values. And thus the search procedure would have to be run numerous times, which can become rather inefficient. Penalization tends to be faster, but does 
require the specification of a penalty function or constant. This choice is highly dependent upon the application field and will require some sort of knowledge about the data. 
Some ways to choose these parameters include cross validation \citep{arlot11}, generalized degrees of freedom \citep{shen02}, and slope heuristics \citep{arlot09}. On the other 
hand, generating all optimal segmentations avoids having to make such a selection.

In the following sections we introduce a change point search procedure which we call cp3o (\textbf{C}hange \textbf{P}oints via \textbf{P}robabilisticly 
\textbf{P}runed \textbf{O}bjectives). This is an exact search procedure that can be applied to a larger number of goodness of fit measures in order to reduce the amount of time 
taken to estimate change point locations. Additionally, the cp3o algorithm allows for the number of change points to be quickly determined without having to specify a penalty 
term, while at the same time generating all other optimal segmentations as a byproduct.\par

As the cp3o procedure can be applied to a general goodness of fit measure we propose one that is based on E-Statistics \cite{rizzo05}, which we call e-cp3o. The e-cp3o method is
a nonparametric algorithm that has the ability to detect \textit{any} type of distributional change. The use of E-Statistics also allows the e-cp3o algorithm to perform analysis 
on multivariate time series without suffering from the curse of dimensionality.\par

The results from a variety of simulations show that our method makes a reasonable trade off between speed and accuracy in most cases.  
In addition to the computational benefits, we show that the cp3o procedure generates consistent estimates for both the number and location of change points when equipped with 
an appropriate goodness of fit measure. Furthermore, under additional assumptions we also show consistency in the setting where the number of change points is increasing 
with the length of time series.\par

The remainder of this paper is organized as follows. In Section \ref{Sec:prob-prune} we discuss the probabilistic pruning procedure used by cp3o, along with 
conditions necessary to ensure consistency. Section \ref{Sec:ecp3o} is devoted to the development of the e-cp3o algorithm and showing that it satisfies the conditions outlined 
in Section \ref{Sec:prob-prune}. Results for applications to both simulated and real datasets are given in Sections \ref{Sec:Simulations} and \ref{Sec:Real}. Concluding remarks 
are left for Section \ref{Sec:Conclusion}.

\section{Probabilistic Pruning}\label{Sec:prob-prune}
When performing change point analysis one must have a quantifiable way of determining whether one segmentation is better than another. When using an exact search procedure this 
is most commonly accomplished through the use of a goodness of fit measure. Suppose that there are $k$ change points $0=\tau_0<\tau_1<\cdots<\tau_k<\tau_{k+1}=T$. These $k$ 
locations partition the time series into $k+1$ segments $C_j=\{Z_i: \tau_{j-1}<i\le\tau_j\}$. The challenge now is to select the change point locations so that the 
observations within each segment are identically distributed, and the distribution of observations in adjacent segments are different. Therefore, we will consider sample goodness 
of fit measures of the following form,
\begin{equation}\label{gof}
	\widehat G(k) = \max_{\tau_1<\tau_2<\cdots<\tau_k}\sum_{j=1}^k\widehat R(C_j, C_{j+1}),
\end{equation}
in which $\widehat R(A,B)$ is a measure of the sample divergence between observation sets $A$ and $B$. The divergence measure $\widehat R$ is such that larger values 
indicate that the distributions of the observations in the two sets are more distinct. Since each term of the sum in Equation \ref{gof} depends only upon contiguous observations 
it is possible to obtain the value of $\widehat G(k)$ through dynamic programming.\par

Using traditional dynamic programming approaches greatly reduces the computational time required to perform the optimization in Equation \ref{gof}. However, the running time of 
such methods is still quadratic in the length of the time series, thus limiting their applicability. Many of the calculations performed during the dynamic programs do not result 
in the identification of a new change point. These calculations can be viewed as excessive since they do not provide any additional information about the series' segmentation, 
and quickly compound to slow down the algorithm. Thus a practical step towards reducing running time, and even possibly the theoretical computational complexity, is to quickly 
identify such excessive calculations and have them removed. One way to do this is by continually pruning the set of potential change point locations. \cite{rigaill10} proposes a 
pruning method that can be used when the goodness of fit measure is convex, and can also be adapted for online change point detection. Since the sample divergence measure 
$\widehat R$ is not necessarily convex this pruning approach may not always be applicable. The cp3o procedure therefore performs pruning that is more in line with 
the approach taken by \cite{killick12} in developing the PELT method.\par

Let $0<v<t<s<u\le T$ and $Z_a^b=\{Z_a, Z_{a+1},\dots, Z_b\}$ for $a\le b$. Furthermore, suppose that there exists a constant $\Gamma$ such that 
$$\widehat R(Z_{v+1}^t,Z_{t+1}^u) - \widehat R(Z_{v+1}^t,Z_{t+1}^s) - \widehat R(Z_{t+1}^s,Z_{s+1}^u) < \Gamma$$
holds for all $v<t<s<u$. The value of $\Gamma$ will depend not only on the distribution of our observations, but also the nature of the divergence measure $\widehat R$. Therefore, 
in many settings it may be difficult, if not impossible, to find such a $\Gamma$. Instead we consider 
the following probabilistic formulation. Let $\epsilon>0$, we then wish to find $\Gamma_\epsilon$ such that for all $v<t<s<u$
$$\mathbb P\!\left(\widehat R(Z_{v+1}^t,Z_{t+1}^u) - \widehat R(Z_{v+1}^t,Z_{t+1}^s) - \widehat R(Z_{t+1}^s,Z_{s+1}^u) \ge \Gamma_\epsilon\right)\le\epsilon.$$
Let $\zeta_k(t)$ denote the value of $\widehat G(k)$ when segmenting $Z_1, Z_2,\dots, Z_t$ with $k$ change points. Using this notation we can express our probabilistic 
pruning rule as follows.

\begin{lem}\label{lemPrune}
	Let $v$ be the optimal change point location preceding $t$, with $t<s<u\le T$. If 
	$$\zeta_k(t)+\widehat R(Z_{v+1}^t,Z_{t+1}^s)+\Gamma_\epsilon < \zeta_k(s),$$ 
	then with probability at least $1-\epsilon$, $t$ is not the optimal change point location preceding $u$.
\end{lem}
\begin{proof}
	If
	\begin{eqnarray*}
		\zeta_k(t)+\widehat R(Z_{v+1}^t,Z_{t+1}^s)+\widehat R(Z_{t+1}^s,Z_{s+1}^u)+\Gamma_\epsilon &<& \zeta_k(s)+\widehat R(Z_{t+1}^s,Z_{s+1}^u),
	\end{eqnarray*}
	then from the definition of $\Gamma_\epsilon$ we have that 
	$$\mathbb P\!\left(\zeta_k(t)+\widehat R(Z_{v+1}^t,Z_{t+1}^u) < \zeta_k(s)+\widehat R(Z_{t+1}^s,Z_{s+1}^u)\right)\ge 1-\epsilon.$$
	Since the optimal value attained from segmenting $Z_1, Z_2,\dots, Z_u$ with $k+1$ change points is an upper bound for $\zeta_k(t)+\widehat R(Z_{t+1}^s,Z_{s+1}^u)$,
	$$\mathbb P\!\left(\zeta_k(t)+\widehat R(Z_{v+1}^t,Z_{t+1}^u) < \zeta_{k+1}(u)\right)\ge 1-\epsilon.$$
	From this we can see that with probability at least $1-\epsilon$, it would be better to have $s$ as the change point prior to $u$.
\end{proof}

\subsection{Consistency}\label{Sec:Consistency}
As has been mentioned before, when performing multiple change point analysis it is of utmost importance to obtain an accurate estimate of the number of change points, as well 
as their locations. Therefore, in this section we will show that under a certain asymptotic setting, the estimates generated by maximizing Equation \ref{gof} generate 
consistent location estimates.\par

When showing consistency many authors consider the case in which the number of change points is held constant, while the number of observations tends toward infinity. This seems 
rather unrealistic, as one would expect to observe additional change points as more data is collected. For this reason we will allow the number of change points, $k(T)$, to 
possibly tend towards infinity as the length of the time series increases. The asymptotic setting we will consider is similar in nature to that taken by \cite{venkatraman92} 
and \cite{zou14}.\par


In order to establish consistency of the proposed estimators we make the following assumptions.
\begin{asm}\label{asm0}
	Let ${\mathscr F}=\{F_0, F_1, \dots\}$ be a collection of distribution functions, and $\{R_{j\ell}\}$ a collection of doubly indexed positive finite constants. Suppose that $A$ 
	and $B$ are disjoint sets of observations, such that the observations in $A$ have distribution $F_a$ and those of $B$ hare distribution $F_b$, for $F_a, F_b\in{\mathscr F}.$ 
	The constants $R_{j\ell}$ are such that $\hat R_{ab}={\widehat R}(A,B)\to R_{ab}$ almost surely as $\min(\#A, \#B)\to\infty$. Furthermore let $f$ be a function such that 
	$|\hat R_{ab}-R_{ab}| = o(f(\#A\wedge\#B))$ almost surely for all pairs $a$ and $b$.
\end{asm}

\begin{asm}\label{asm1}
	Let $\displaystyle\lambda_T=\min_{1\le j\le k(T)}(\tau_j-\tau_{j-1})$, and suppose $\lambda_T\to\infty$ as $T\to\infty$. 
\end{asm}
Assumption \ref{asm1} states that the number of observations between change points tends towards infinity. This later allows us to apply the law of large numbers.
\begin{asm}\label{asm2}
	The number of change points $k(T)$ and its upper bound $K(T)$ are such that $k(T)=o\!\left(\frac{1}{f(T)}\right)$ and $k(T)=o(K(T))$.
\end{asm}
The above assumption controls the rate at which the number of change points can increase. This is directly related to the rate at which our sample estimates 
converge almost surely to their population counterparts.
\begin{asm}\label{asm3}
	Let $\mathscr F$ be the collection of distribution functions from Assumption \ref{asm0}. From this collection we define a set of random variables 
	$\{X(a,b)\}_{a,b=0}^\infty.$ For each pair of values $0\le a\le b$ the random variable $X(a,b)$ has a mixture distribution created with mixture components 
	$F_a, F_{a+1},\dots, F_b.$\par
	Then for $r\le q$, and integers $0=s_0<s_1<s_2<\cdots<s_r<s_{r+1}=q+1$, define,
	\begin{equation*}
		G^q(r) = \max_{s_1, s_2,\dots, s_r}\sum_{i=1}^rR(X(s_{i-1},s_i-1), X(s_i,s_{i+1}-1));
	\end{equation*}
	and for $r>q$ we define $G^q(r)$ to be equal to $G^q(q)$. 
	Assume that $\displaystyle d_T=\max_{1\le i,j\le k(T)}G^{k(T)}(i)-G^{k(T)}(j),$ is such that 
	$[d_Tk(T)]/[K(T)]\to 0$ as $T\to\infty$.
	
\end{asm}
Assumption \ref{asm3} concerns the rate at which additional change points increase the objective function of interest. We will show that a higher upper bound on the 
number of change points is necessary when each additional change point has the potential to greatly change the value of $G^{q}(r)$.
\begin{asm}\label{asm4}

	
	Let $0\le\pi<\gamma<\rho\le1$ and $i_1<i_2<i_3<i_4$ be positive integers. Suppose that the time series $Z_1, Z_2, \dots, Z_T$ is such that 
	$Z_{\lfloor\pi T\rfloor+1},\dots,Z_{\lfloor\gamma T\rfloor}\sim X(i_1,i_2)$ and $Z_{\lfloor\gamma T\rfloor+1},\dots,Z_{\lfloor\rho T\rfloor}\sim X(i_3,i_4)$ for every 
	sample of size $T$. For $\tilde\gamma\in(\pi,\rho)$ we define the following sets $A(\tilde\gamma)=\{Z_{\lfloor\pi T\rfloor},\dots, Z_{\lfloor\tilde\gamma T\rfloor}\}$ 
	and $B(\tilde\gamma) = \{Z_{\lfloor\tilde\gamma T\rfloor+1},\dots,Z_{\lfloor\rho T\rfloor}\}$. Assume that there exist a class of functions indexed by $\pi, \gamma$, and 
	$\rho$; $\Theta_\pi^\rho(\tilde\gamma|\gamma):(\pi,\rho)\mapsto\mathbb R,$ such that 
	$\widehat R(A(\tilde\gamma),B(\tilde\gamma))\to \Theta_\pi^\rho(\tilde\gamma|\gamma)R(X(i_1,i_2),X(i_3,i_4))$ almost surely as $T\to\infty$. Finally we assume that 
	$\Theta_\pi^\rho(\tilde\gamma|\gamma)$ has a unique maximizer at $\tilde\gamma=\gamma$.
\end{asm}
Assumption \ref{asm4} describes the behavior of our goodness of fit measure when it is used to identify a single change point. Essentially this assumption states that the 
measure will attain its maximum value when the estimated change point location, $\tilde\gamma$, and true change point location, $\gamma$, coincide.\par

\subsubsection*{Change Point Locations}
We begin by showing that under Assumptions \ref{asm0}-\ref{asm4}, the cp3o procedure will produce consistent estimates for the change point locations.

\begin{lem}\label{lem1}
    Let $\mathcal{G}(k(T)) = \{\hat\tau_1, \hat\tau_2,\dots, \hat\tau_{k(T)}\}$ and 
    \begin{equation*}
		\mathcal{B}_T(\epsilon) = \mathcal{B}_T(\epsilon,\{\tau_i\})
    = \left\{(\eta_1, \eta_2, \dots, \eta_{k(T)})\in\mathbb R^{k(T)}: \frac{|\eta_i-\tau_i|}{T}\le\epsilon\ \mbox{for }i=1, 2, \dots, k(T)\right\}.
    \end{equation*}
    Then for all $\epsilon>0$,
    $$\mathbb P\left(\mathcal{G}(k(T))\in\mathcal{B}_T(\epsilon)\right)\to 1$$
    as $T\to\infty$.
\end{lem}

\begin{proof}
    Suppose that $\mathcal{G}(k(T))\notin\mathcal{B}_T(\epsilon)$, then there exists $i$ such that $\frac{|\hat\tau_i-\tau_i|}{T}>\epsilon.$ Select the largest such $i$ and 
    define the following random variables. Let $M_1\sim U_1$ where $U_1$ is the distribution (possibly a mixture) created by the observations between $\hat\tau_{i-2}$ and 
    $\hat\tau_{i-1}$, $M_2\sim U_2$ for $U_2$ having distribution created by the observations between $\hat\tau_{i-1}$ and $\tau_i$. Similarly define $M_3$ for the observations 
    between $\tau_i$ and $\hat\tau_{i+1}$, and $M_4$ for the observations between $\hat\tau_{i+1}$ and $\hat\tau_{i+2}$.\par
    
    Then the value of the sample goodness of fit measure generated by the estimates of $\mathcal G(k(T))$ is
    \begin{equation*}
    	\hat R(M_1,M_2) +  \hat R(M_2,M_3) + \hat R(M_3,M_4) + A,
    \end{equation*}
    which due to Assumptions \ref{asm0} and \ref{asm2} is equal to 
    \begin{equation*}
    	\Theta_0^1(\beta_1|\gamma_1)R(U_1,U_2) +  \Theta_0^1(\beta_2|\gamma_2)R(U_2,U_3) + \Theta_0^1(\beta_3|\gamma_3)R(U_3,U_4) + B + k(T)o(f(T)).
    \end{equation*}
    In the above expressions $A$ and $B$ are collections of terms that are not affected by the choice of $\hat\tau_i$. The $\beta_i$ and $\gamma_i$ terms are as listed below.
\begin{center}
    $\begin{array}{ccc}
    	\beta_1=\frac{\hat\tau_{i-1}-\hat\tau_{i-2}}{\hat\tau_i-\hat\tau_{i-2}} & \beta_2=\frac{\hat\tau_i-\hat\tau_{i-1}}{\hat\tau_{i+1}-\hat\tau_{i-1}} & 
    	\beta_3 = \frac{\hat\tau_{i+1}-\hat\tau_i}{\hat\tau_{i+2}-\hat\tau_i}\\
    	\\
    	\gamma_1=\frac{\hat\tau_{i-1}-\hat\tau_{i-2}}{\tau_i-\hat\tau_{i-2}} & \gamma_2=\frac{\tau_{i}-\hat\tau_{i-1}}{\hat\tau_{i+1}-\hat\tau_{i-1}}  &
    	\gamma_3=\frac{\hat\tau_{i+1}-\tau_i}{\hat\tau_{i+2}-\tau_i}\\
    \end{array}$
\end{center}
    Each of the terms in the sum is maximized when $\beta_i=\gamma_i$, 
    which corresponds to $\frac{|\hat\tau_i-\tau_i|}{T}\to 0.$ By our assumptions, we have that the remainder term $k(T)o(f(T)) = o(1).$ Therefore, if 
    $\frac{|\hat\tau_i-\tau_i|}{T}$ is strictly bounded away from $0$ then the statistic will be strictly less than the optimal value as $T\to\infty$. However, this contradicts 
    the manner in which $\hat\tau_i$ is selected.
\end{proof}

\subsubsection*{Number of Change Points}
Once we have shown that the procedure generates consistent estimate for the change point locations it is simple to show that it will also produce a consistent estimate for 
the number of change points. We have chosen to implement the procedure outlined below in Assumption \ref{asm5} to determine the number of change points. However, 
other approaches could be used and still have the same consistency result.\par

\begin{asm}\label{asm5}
 Define $\nabla\widehat G(k) = \widehat G(k+1) - \widehat G(k)$, 
 $\widehat\mu(\nabla\widehat G)=\frac{\widehat G(K(T)) - \widehat G(1)}{K(T) - 1}$ and 
 $\widehat\sigma^2(\nabla\widehat G)=\widehat{Var}\{\nabla\widehat G(k): k=1,2,\dots, K(T)-1\}$.
 Then suppose our estimated number of change points is given by
 \begin{equation}\label{khat}
     \hat k(T) = 1 + \max\left\{\ell: \nabla\widehat G(1),\dots, \nabla\widehat G(\ell)> \widehat\mu(\nabla\widehat G)+
     \frac{1}{2}\sqrt{\widehat\sigma^2(\nabla\widehat G)}\right\}.
 \end{equation}
\end{asm}
The selection procedure in Equation \ref{khat} has similar intuition to the one presented by \cite{lavielle05}. Both procedures work on the principle that a true change point 
will cause a significant change in the goodness of fit. While spurious change point estimates will only cause a minuscule increases/decrease in value. We thus say 
that a change is significant if it is more than half a standard deviation above the average change. As previously stated, other methods could be used, this just happens to be the 
one that we chose to implement.\par

Before proving that we can obtain a consistent estimate for the number of change points one final assumption is made to ensure that the detection of an 
additional true change point causes a strictly positive increase in our goodness of fit measure. A similar property is also needed for the finite sample approximation.
\begin{asm}\label{asm6}
	For every fixed $q$, suppose the values $G^q(1), G^q(2),\dots, G^q(q)$ form an increasing sequence. Similarly let $\hat G^q(r)$ be the finite sample estimates of $G^q(r)$. 
	Additionally, suppose there exists $T_0$ such that for $T>T_0$ the values $\hat G^{k(T)}(1), \hat G^{k(T)}(2),\dots, \hat G^{k(T)}(k(T))$ also form an increasing sequence.
\end{asm}
\begin{lem}\label{cpNumber}
    Let $\hat k(T)$ be the number of estimated change points for a sample of size $T$, and that the conditions of Assumptions \ref{asm0}-\ref{asm6} hold. Then
    \begin{equation*}
        \lim_{T\to\infty}{\mathbb P}\left(\hat k(T) = k(T)\right) = 1.
    \end{equation*}
\end{lem}
\begin{proof}
	If $k(T)$ is bounded then the proof for the constant $k(T)$ version applies. Suppose that $\hat k(T)>k(T)$, then $\nabla\hat G^{k(T)}(k(T))>\mu(\nabla\hat G)+
    \frac{1}{2}\sqrt{\sigma^2(\nabla\hat G)}$. Letting $\nabla_{ij}^T=\frac{1}{2}[G^{k(T)}(i)-G^{k(T)}(j)]^2,$ we note the following inequalities:
    \begin{eqnarray*}
        \mu(\nabla\hat G)&=&\frac{\hat G^{k(T)}(K(T))-\hat G^{k(T)}(1)}{K(T)-1}\\
        &\le& \frac{d_T}{K(T)-1}+o(1)\\
        &\to& 0.
    \end{eqnarray*}
	\begin{eqnarray*}
	    \sigma^2(\nabla\hat G)&=&\frac{2}{(K(T)-1)(K(T)-2)}\sum_{1\le i<j\le K(T)-1}\frac{1}{2}\left[G^{k(T)}(j)-G^{k(T)}(i)\right]^2 + o(1)\\
	    &=&\binom{K(T)-1}{2}^{-1}\left[\sum_{1\le i<j\le k(T)}\hspace{-3ex}\nabla_{ij}^T \hspace{3ex} +
			\mathop{\sum_{1\le i\le k(T)<j\le K(T)-1}}\hspace{-6ex}\nabla_{ij}^T \hspace{4ex}+ 
			\sum_{k(T)<i<j\le K(T)-1}\hspace{-5ex}\nabla_{ij}^T\hspace{2ex}\right] + o(1)\\
	    &\le&\binom{K(T)}{2}^{-1}\left[\binom{k(T)}{2}d_T + k(T)(K(T)-k(T))d_T + \binom{K(T)-k(T)}{2}o(1)\right] + o(1)\\
	    &\to&0,
	\end{eqnarray*}
since $G^{k(T)}(i)=G^{k(T)}(k(T))$ for $i>k(T)$. Thus $\mathbb P(\hat k(T)>k(T))\to 0$ as $T\to\infty$.\par

Next suppose that $\hat k(T)<k(T)$. This implies that $\nabla\hat G^{k(T)}(\hat k(T))\le \mu(\nabla\hat G)+\frac{1}{2}\sqrt{\sigma^2(\nabla\hat G)}$. However, since 
$\nabla G^{k(T)}(i)>0$ for $i<k(T)$, $\mathbb P(\hat k(T)<k(T))\to 0$ as $T\to\infty$.    
\end{proof}

\begin{thm}\label{thm1}
    For all $\epsilon>0$, as $T\to\infty$, 
    $$\mathbb P\left(\hat k(T)=k(T), \mathcal{G}(k(T))\in\mathcal{B}_T(\epsilon)\right)\to 1.$$
\end{thm}
\begin{proof}
    Using the results of Lemma \ref{lem1} and Lemma \ref{cpNumber} we have the following;
    \begin{eqnarray*}
        \mathbb P\left(\hat k(T)=k(T), \mathcal{G}(k(T))\in\mathcal{B}_T(\epsilon)\right) &  \ge & 1 - \mathbb P\left(\hat k(T)\neq k(T)\right) - 
																									\mathbb P\left(\mathcal{G}(k(T))\notin \mathcal{B}_T(\epsilon)\right)\\
		& \to & 1.
    \end{eqnarray*}
\end{proof}

\section{Pruning and Energy Statistics}\label{Sec:ecp3o}
The cp3o procedure introduced in Section \ref{Sec:prob-prune} can be applied with almost any goodness of fit measure $\widehat R(\cdot,\cdot)$. However, in order to ensure 
consistency for both the estimated change point locations, as well as the estimated number of change points, some restrictions must be enforced as outlined in Section 
\ref{Sec:Consistency}.\par

In this section we make use of a particular class of goodness of fit measures that allows for nonparametric change point analysis. These measures are indexed by $\alpha\in(0,2)
\footnote{The choice of $\alpha=2$ is allowed, however, in this case the goodness of fit measure would only be able to detect changes in mean.}$
and allow for the detection of \textit{any} type of distributional change. When a value of $\alpha$ is selected, the only distributional assumptions that are made are that 
observations are independent and that they all have finite absolute $\alpha$th absolute moments. This class of measures are based upon the energy statistic of \cite{rizzo05}, 
and we thus call the resulting procedure e-cp3o.\par

The e-cp3o procedure is a nonparametric procedure that makes use of an approximate test statistic and an exact search algorithm in order to locate change points. Computationally 
the e-cp3o procedure is comparable to other parametric/nonparametric change point methodologies that use approximate search algorithms. In the remainder of this section we 
give a brief review of E-Statistics, followed by their incorporation into the cp3o framework. Finally, we show that the resulting goodness of fit measure satisfies the conditions 
necessary for consistency.\par

\subsection{The Energy Statistic}\label{eStat}
As change point analysis is directly related to the detection of differences in distribution we consider the U-statistic introduced in \cite{rizzo05}. This statistic provides a 
simple way to determine whether the independent observations in two sets are identically distributed.\par

Suppose that we are given samples $\X_n=\{X_i: i=1,\dots, n\}$ and $\Y_m=\{Y_j: j=1,\dots, m\}$, that are independent iid samples from distributions $F_X$ and $F_Y$ respectively. 
Our goal is to determine if $F_X=F_Y$. We then define the following metric on the space of characteristic functions,
$$\mathcal D(X,Y|\alpha)=\int_{\mathbb R^d}|\phi_x(t)-\phi_y(t)|^2\omega(t|\alpha)\,dt,$$
i which $\phi_x$ and $\phi_y$ are the characteristic functions associated with distributions $F_X$ and $F_Y$ respectively. Also $\omega(t|\alpha)$ is a positive weight function 
chosen such that the integral is finite. By the uniqueness of characteristic functions, it is obvious that $\mathcal D(X,Y|\alpha)=0$ if and only if $F_X=F_Y$.\par

Another metric that can be considered is based on Euclidean distances. Let $(X',Y')$ be an iid copy of $(X,Y)$, then for $\alpha\in(0,2)$ define
\begin{equation}\label{energyStat}
    \mathcal E(X,Y|\alpha)=2E|X-Y|^\alpha - E|X-X'|^\alpha - E|Y-Y'|^\alpha.
\end{equation}
For an appropriately chosen weight function, 
$$\omega(t|\alpha)=\left(\frac{2\pi^{d/2}\Gamma(1-\alpha/2)}{\alpha 2^\alpha\Gamma((d+\alpha)/2)}|t|^{d+\alpha}\right)^{-1},$$
we have the following lemma.

\begin{lem}\label{metrics}
    For any pair of independent random variables $X$ and $Y$, and $\alpha\in(0,2)$ is such that $E(|X|^\alpha+|Y|^\alpha)<\infty$, then 
    $\mathcal E(X,Y|\alpha)=\mathcal D(X,Y|\alpha)$ 
    and $\mathcal E(X,Y|\alpha)\in[0,\infty)$. Moreover, $\mathcal E(X,Y|\alpha)=0$ if and only if $X$ and $Y$ are identically distributed.
\end{lem}
\begin{proof}
    See the appendix of \cite{rizzo05}.
\end{proof}

Theorem \ref{metrics} allows  for an intuitively simple empirical divergence measure. Let $\X_n$ and $\Y_m$ be as above, then we can define the empirical counterpart to 
Equation \ref{energyStat}
\begin{equation}\label{energyStatEmp}
    \widehat{\mathcal E}(\X_n,\Y_m;\alpha)=\frac{2}{mn}\sum_{i=1}^n\sum_{j=1}^m|X_i-Y_j|^\alpha - \binom{n}{2}^{-1}\!\!\!\sum_{1\le i<j\le n}|X_i-X_j|^\alpha 
    -\binom{m}{2}^{-1}\!\!\!\sum_{1\le i<j\le m}|Y_i-Y_j|^\alpha.
\end{equation}

\subsection{Incomplete Energy Statistic}\label{incompleteStatistic}
The computation of the U-statistics presented in Equation \ref{energyStatEmp} require $\mathcal O(n^2\vee m^2)$ calculations, which makes it impractical for large $n$ or $m$. 
We propose working with an approximate statistic that is obtained by using incomplete U-statistics. In the following formulation of the incomplete U-statistic let 
$\delta\in\{2,3,\dots,\lfloor\sqrt{T}\rfloor\}$.\par

Suppose that we divide a segment of our time series into two adjacent subseries, $\X_n=\{Z_a, Z_{a+1},\dots, Z_{a+n-1}\}$ and $\Y_m=\{Z_{a+n}, Z_{a+n+1},\dots, Z_{a+n+m-1}\},$ 
and define the following sets
\begin{eqnarray*}
	W_X^\delta & = &\{(i,j): a+n-\delta\le i<j<a+n\}\cup\bigcup_{i=0}^{n-\delta-1}\{(a+i,a+i+1)\}\\
	W_Y^\delta & = &\{(i,j): a+n\le i<j<a+n+\delta\}\cup\bigcup_{i=\delta-1}^{m-2}\{(a+n+i,a+n+i+1)\}\\
	B^\delta & = & \left(\{a+n-1,\dots,a+n-\delta\}\times\{a+n,\dots,a+n+\delta-1\}\right)\cup
					\left(\bigcup_{i=\delta+1}^{m\wedge n}\{(a+n-i,a+n+i-1)\}\right)
\end{eqnarray*}
The set $B^\delta$ aims at reducing the number of samples needed to compute the between sample distances. While the sets $W_X^\delta$ and $W_Y^\delta$ reduce the number of 
terms used for the within sample distances. When making this reduction the sets $W_X^\delta, W_Y^\delta,$ and $B^\delta$ consider all unique pairs within a $\delta$ window 
around the split that creates $\X_n$ and $\Y_m$. This point corresponds to a potential change point location and thus we use as much information about points close by to 
determine the empirical divergence.\par

We then define the incomplete U-statistic $\widetilde{\mathcal E}$ as
\begin{equation}\label{incUstat}
\widetilde{\mathcal E}(\X_n,\Y_m|\alpha,\delta) = \frac{2}{\#B^\delta}\sum_{(i,j)\in B^\delta}|X_i-Y_j|^\alpha - \frac{1}{\#W_X^\delta}\sum_{(i,j)\in W_X^\delta}|X_i-X_j|^\alpha
- \frac{1}{\#W_Y^\delta}\sum_{(i,j)\in W_Y^\delta}|Y_i-Y_j|^\alpha.
\end{equation}
Using this approximation greatly reduces our computational complexity from $\mathcal O(n^2\vee m^2)$ to $\mathcal O(n\vee m)$. \cite{nasari12} shows that a strong law of large 
numbers result holds for incomplete U-Statistics, and and thus $\widehat{\mathcal E}$ and $\widetilde{\mathcal E}$ have the same almost sure limit as $n\wedge m\to\infty$.\par

\subsection{The e-cp3o Algorithm}
We now present the goodness of fit measure that is used by the e-cp3o change point procedure. In addition, we show that the prescribed measure satisfies the necessary 
consistency requirements from Section \ref{Sec:Consistency}.\par

The e-cp3o algorithm uses an approximate test statistics combined with an exact search algorithm in order to identify change points. Its goodness of fit measure is given by 
the following weighted U-Statistic,
\begin{equation}\label{Rhat}
    \widehat{\mathcal R}(\X_n,\Y_m|\alpha)=\frac{mn}{(m+n)^2}\widehat{\mathcal E}(\X_n,\Y_m|\alpha).
\end{equation}
Or an approximation can be obtained by using its incomplete counterpart
\begin{equation*}
    \widetilde{\mathcal R}(\X_n,\Y_m|\alpha,\delta)=\frac{mn}{(m+n)^2}\widetilde{\mathcal E}(\X_n,\Y_m|\alpha,\delta).
\end{equation*}
By using Slutsky's theorem and a result of \cite{rizzo10} we have that if $F_X=F_Y$ then $\widehat{\mathcal R}(\X_n,\Y_m|\alpha)\stackrel{p}{\to}0$, and otherwise 
$\widehat{\mathcal R}(\X_n,\Y_m|\alpha)$ tends almost surely to a finite positive constant, provided that $m=\Theta(n)$ (this means that $n=\mathcal O(m)$ and 
$m=\mathcal O(n)$). In fact if $F_X=F_Y$ we have that $\widehat{\mathcal R}(\X_n,\Y_m|\alpha,\delta)\to 0$ almost surely.\par

In the case of $\widetilde{\mathcal R}(\X_n, \Y_m|\alpha)$, the result of \cite[Theorem 4.1]{oneil93} combined and Slutsky's theorem show that under equal distributions, 
$F_X=F_Y$, $\widetilde{\mathcal R}(\X_n, \Y_m|\alpha,\delta)\stackrel{p}{\to}0.$ Similarly, $\widehat{\mathcal R}(\X_n,\Y_m|\alpha,\delta)$ also tends towards 
a positive finite constant provided $m=\Theta(n).$ These properties lead to a very intuitive goodness of fit measure,
\begin{equation*}
    \widehat{G}(k|\alpha)=\max_{\tau_1<\tau_2<\dots<\tau_k}\ \sum_{j=1}^{k}\widehat{\mathcal R}(C_j,C_{j+1}|\alpha).
\end{equation*}
By using the dynamic programming approach presented by \cite{lavielle06}, the values $\widehat{G}(\ell|\alpha)$ for $\ell\le k$ can be computed in $\mathcal O(kT^3)$ 
instead of $\mathcal O(T^{k+2})$ operations. However, the $T^3$ term makes this an inadequate approach, so the procedure (e-cp3o) is implemented with the similarly 
defined goodness of fit measure $\widetilde{G}(k|\alpha,\delta)$, which allows for only $\mathcal O(kT^2)$ operations.\par

\subsubsection*{Consistency of e-cp3o}
We now show that the goodness of fit measure, $\widetilde R(\cdot,\cdot|\alpha,\delta)$, used by e-cp3o satisfies the conditions for a cp3o based procedure to generate 
consistent estimates. It is assumed that $\alpha$ has been chosen so that all of the $\alpha$th moments are finite. In the results below we will consider the goodness of fit 
measure $\widehat R$ based upon the complete U-Statistic, even though the e-cp3o procedure is based on its incomplete version $\widetilde R$. The reason for this is that 
$\widehat R$ and $\widetilde R$ have the same almost sure limits, and we are working in an asymptotic setting.
\begin{prop}
	Assumption \ref{asm4} is satisfied by the e-cp3o goodness of fit measure.
\end{prop}
\begin{proof}
	Using the result of \cite[Theorem 1]{matteson13} we have that 
	$$\widehat R(A(\tilde\gamma), B(\tilde\gamma)|\alpha)\to \tilde\gamma(1-\tilde\gamma)h(\tilde\gamma;\gamma){\mathcal E}(U_1, U_2|\alpha).$$
	Such that $U_1\sim X(i_1,i_2)$ $U_2\sim X(i_3,i_4),$ and $h(x;y) =\left(\frac{y}{x}\mathbbm{1}_{x\ge y} + \frac{1-y}{1-x}\mathbbm{1}_{x<y}\right)^2.$ Therefore, 
	$R(X(i_1,i_2),X(i_3,i_4))=\gamma(1-\gamma){\mathcal E}(U_1,U_2|\alpha)$ and $\Theta_0^1(\tilde\gamma|\gamma)=\frac{\tilde\gamma(1-\tilde\gamma)}{\gamma(1-\gamma)}h(\tilde\gamma;\gamma)$, 
	which can be shown to have a unique maximizer at $\tilde\gamma=\gamma$.
\end{proof}

\begin{prop}
	The portion of Assumption \ref{asm6} about $\left\{G^m(r)\right\}_{r=1}^m$ holds for the e-cp3o goodness of fit measure.
\end{prop}
\begin{proof}    
    We begin by showing that $G^m(1)<G^m(2)$. Suppose the first change point partitions the time 
    series into two segments, one where observations are distributed according to $F$ and another where they are distributed according to $J$. Now suppose that $J$ is a 
    created by a linear mixture of the distributions $G$ and $H$ (which may themselves be mixture distributions). Suppose that the second change point is positioned so as 
    to separate these distributions, $G$ and $H.$ Let random variables $X, Y,$ and $Z$ be such that $X\sim F, Y\sim G,$ and $Z\sim H$. Then we have that
    \begin{equation*}
        G^m(1) = \int_{\mathbb R^d}|\phi_x(t)-\phi_y(\beta t)\phi_z((1-\beta)t)|^2\omega(t|\alpha)\,dt
    \end{equation*}
	where $\beta$ is the mixture coefficient used to create the distribution $J$. It is clear that this will be maximized either when $\beta=0$ or $\beta=1$, in either case we 
	will show that the obtained value is bounded above by
	\begin{equation*}
	    G^m(2)=\int_{\mathbb R^d}|\phi_x(t)-\phi_y(t)|^2\omega(t|\alpha)\,dt + \int_{\mathbb R^d}|\phi_y(t)-\phi_z(t)|^2\omega(t|\alpha)\,dt.
	\end{equation*}
	
	\noindent\underline{Case $\beta=1$:} In this setting the value of $G^m(1)$ is equal to the first term in the definition of $G^m(2)$, and since the 
	distributions $G$ and $H$ are distinct, the second term is strictly positive. Thus $G^m(1)<G^m(2)$.	
	
	\noindent\underline{Case $\beta=0$:} In this case $G^m(1)=\int|\phi_x(t)-\phi_z(t)|^2\omega(t|\alpha)dt$. However, since we have a metric, the 
	triangle inequality immediately shows that $G^m(1)<G^m(2)$.\par
	
	In the above setting the location of the first change point was held fixed when the second was identified. This need not be the optimal way to partition the time series into 
	three segments. Thus since this potentially suboptimal segmentation results in an upper bound for $G^m(1)$ it follows that the optimal segmentation will 
	also bound $G^m(1)$.\par
	
	The argument to show that $G^m(r)<G^m(r+1)$ for $r=2,\dots, m-1$ is identical.
\end{proof}

\begin{prop}
	The portion of Assumption \ref{asm6} about $\left\{\hat G^{k(T)}(r)\right\}_{r=1}^{k(T)}$ holds for the e-cp3o goodness of fit measure.
\end{prop}
\begin{proof}
	In the paper \cite{rizzo05}, the empirical measure used for 
    the statistic $\widehat{\mathcal E}$ is based upon V-statistics, while we instead use U-statistics. The use of V-statistics ensures that the statistic will always 
    have a nonnegative value. This isn't  the case when using U-statistics, but the difference in their value can be bounded by a constant multiple of $\frac{1}{T}$. 
    Combining this with the fact that $0<G^{k(T)}(r)<G^{k(T)}(r+1)$, and $\frac{d_T}{T}\to 0$, we conclude that for $T$ large enough the version of the statistics 
    based on U-statistics will also produce nonnegative values. Therefore for $T$ large enough $\hat G^{k(T)}(r)<\hat G^{k(T)}(r+1)$.
\end{proof}

\section{Simulation Study}\label{Sec:Simulations}
We now show the effectiveness of our methodology by considering a number of simulation studies. The goal of these studies is to demonstrate that the e-cp3o procedure is able to 
perform reasonably well in a variety of settings. In these studies we examine both the number of estimated change points as well as their estimated locations.\par

To assess the performance of the segmentation obtained from the e-cp3o procedure we use Fowlkes and Mallows' adjusted Rand index \citep{fowlkes83}. This value is calculated by 
comparing a segmentation based upon estimated change point locations to the known true segmentation. The index takes into account both the number of change points as well as their 
locations, and lies in the interval $[0,1]$, where it is equal to $1$ if and only if the two segmentations are identical.\par

For each simulation study we apply various methods to 100 randomly generated time series. We then report the average running time in seconds, the average adjusted Rand value, 
and the average number of estimated change points.\par

As the simulations in the following sections will demonstrate, the e-cp3o procedure does not always generate the best running time or average Rand values. However, 
in every setting it generates results that are either better or comparable to almost all other competitors, when accuracy and speed are viewed together. For 
this reason we would advocate the use of the e-cp3o procedure as a general purpose change point algorithm, especially for small to moderate length time series.\par

To perform the probabilistic pruning introduced in Section \ref{Sec:prob-prune} the value of $\Gamma_\epsilon$ must be specified. In our implementation we obtain an 
estimate of $\Gamma_\epsilon$ in the following way. We uniformly draw $\mathcal O\!\!\left(\frac{1}{\epsilon}\right)$ random samples from the set 
$$\left\{(v,t,s,u): v<t<s<u\mbox{ and }\min\{t-v,s-t,u-s\}\ge\delta\right\}.$$
For each sample we calculate 
$$\widetilde R(Z_{v+1}^t,Z_{t+1}^u;\alpha) - \widetilde R(Z_{v+1}^t,Z_{t+1}^s;\alpha) - \widetilde R(Z_{t+1}^s,Z_{s+1}^u;\alpha),$$
and then set $\Gamma_\epsilon$ equal to the $1-\epsilon$ quantile of these quantities. Any other sampling approach could be used to obtain a value for $\Gamma_\epsilon$ 
as long it satisfies the probabilistic criterion.\par

\subsection{Univariate Simulations}\label{subsec:UniSim}
We begin our simulation study by comparing the e-cp3o procedure to the E-Divisive and PELT procedures. These two procedures are implemented in the \texttt{ecp} \citep{ecp14} and 
\texttt{changepoint} R packages respectively. This set of simulations consist of independent Gaussian observations which undergo changes in their mean and variance. 
The distribution parameters were chosen so that $\mu_j\stackrel{iid}{\sim}U(-10,10)$ and $\sigma^2_j\stackrel{iid}{\sim}U(0,5)$. 
For each analyzed time series all of the different change point procedures were run with their default parameter values. For E-Divisive and e-cp3o this corresponds 
to $\alpha=1$. And for e-cp3o the minimum segment size was set to 30 observations (corresponding to $\delta=29$), and a value of $\epsilon=0.01$ 
is used for the probabilistic pruning. Since in this simulation study the number of change points increased with the time series length, the value of $K(T)$ would also change. 
The results of these simulations are in Table \ref{growthSim}, which also includes additional information about the time series and upper limit $K(T)$.
\begin{table}[t]
\centering
\begin{tabular}{cddd}
\hline
\multicolumn{1}{c}{} & \multicolumn{1}{c}{PELT} & \multicolumn{1}{c}{E-Div} & \multicolumn{1}{c}{e-cp3o}\\ \hline
\multicolumn{1}{c}{} & \multicolumn{3}{c}{T=400, k(T)=3, K(T)=9}\\ \cline{2-4}
Rand                                 & 0.884_{7\e{-3}}& 0.987_{2\e{-3}}& 0.937_{10^{-2}}\\
\# of cps                            & 9.150_{0.4}& 3.070_{3\e{-2}}& 2.660_{8\e{-2}}\\
Time(s)                              & 0.003_{5\e{-5}}& 9.199_{5\e{-2}}& 0.150_{7\e{-4}}\\ \hline

\multicolumn{1}{c}{} & \multicolumn{3}{c}{T=1650, k(T)=10, K(T)=50}                                                                       \\ \cline{2-4}
Rand                                 & 0.953_{3\e{-3}}& 0.992_{7\e{-4}}& 0.940_{5\e{-3}}\\
\# of cps                            & 16.690_{0.4}& 10.050_{2\e{-2}}& 9.390_{7\e{-2}}\\
Time(s)                              & 0.009_{7\e{-5}}& 239.263_{0.7}& 3.542_{3\e{-2}}\\ \hline
\end{tabular}
\caption{\label{growthSim}
Results of the first univariate simulation from Section \ref{subsec:UniSim} with different time series lengths. The true number of change points is given by 
k(T) and K(T) the upper limit used by e-cp3o. The table contains average values over 100 replicates, with standard error as subscripts.
}
\end{table}

As can be seen from Table \ref{growthSim}, better results are obtained by combining an exact test statistic with an approximate search algorithm. But these gains in segmentation 
quality are rather small. Thus, because of the increase in speed and small loss in segmentation quality, we would argue that the e-cp3o procedure should be preferred over the 
E-Divisive. The PELT procedure was much faster, but the e-cp3o procedure was able to generate segmentations that were similar in quality as measured by the 
adjusted Rand index.\par

The next set of simulations also compares to a nonparametric procedure from the \texttt{npcp} R package. This procedure, like the e-cp3o, 
is designed to detect changes in the joint distribution of multivariate time series. More information about this procedure, which we will denote by NPCP-F, is given in Section 
\ref{subsec:MultiSim}. Time series in this simulation study contain two changes in mean followed by a change in tail index. The changes in mean correspond to the data 
transitioning from a standard normal distribution to a $N(3,1)$ and then back to standard normal. The tail index change is caused by a transition to a t-distribution with $2.01$ 
degrees of freedom. We expect that all three methods will be able to easily detect the mean changes and will have a more difficult time detecting the change in tail index. As 
with the previous set of simulations, all procedures are run with their default parameter values. Results for this set of simulations can be found in Table \ref{normSim}. 
Surprisingly, in this set of simulations the e-cp3o procedure was not only significantly faster than the E-Divisive and NPCP-F, but also managed to generate 
slightly better segmentations on average.\par

\begin{minipage}{\linewidth}
	\begin{minipage}{.45\linewidth}
		\begin{table}[H]
\begin{tabular}{lddd}
\hline
\multicolumn{1}{c}{} & \multicolumn{1}{c}{NPCP-F} & \multicolumn{1}{c}{E-Div} & \multicolumn{1}{c}{e-cp3o}\\ \hline
\multicolumn{1}{c}{} & \multicolumn{3}{c}{T=400, k(T)=3, K(T)=9} \\ \cline{2-4}
Rand                                 & 0.820_{5\e{-3}}& 0.828_{5\e{-3}}& 0.874_{7\e{-3}}\\
\# of cps                            & 2.280_{6\e{-2}}& 2.200_{5\e{-2}}& 2.430_{5\e{-2}}\\
Time                                 & 4.790_{2\e{-2}}& 6.726_{5\e{-2}}& 0.176_{10^{-3}}\\ \hline

\multicolumn{1}{c}{} & \multicolumn{3}{c}{T=1600, k(T)=3, K(T)=9}                                                                      \\ \cline{2-4}
Rand                                 & 0.839_{6\e{-3}}& 0.864_{8\e{-3}}& 0.917_{5\e{-3}}\\
\# of cps                            & 2.370_{6\e{-2}}& 2.480_{7\e{-2}}& 2.920_{3\e{-2}}\\ 
Time                                 & 71.772_{0.3}& 143.207_{1.1}& 2.392_{4\e{-2}}\\ \hline
\end{tabular}
\caption{\label{normSim}
Simulation results for time series with mean and tail index changes. The subscripts indicate the standard errors for each value.
}
\end{table}

	\end{minipage}
	\hspace{0.05\linewidth}
	\begin{minipage}{.45\linewidth}
		\begin{figure}[H]
		    \includegraphics[width=\linewidth]{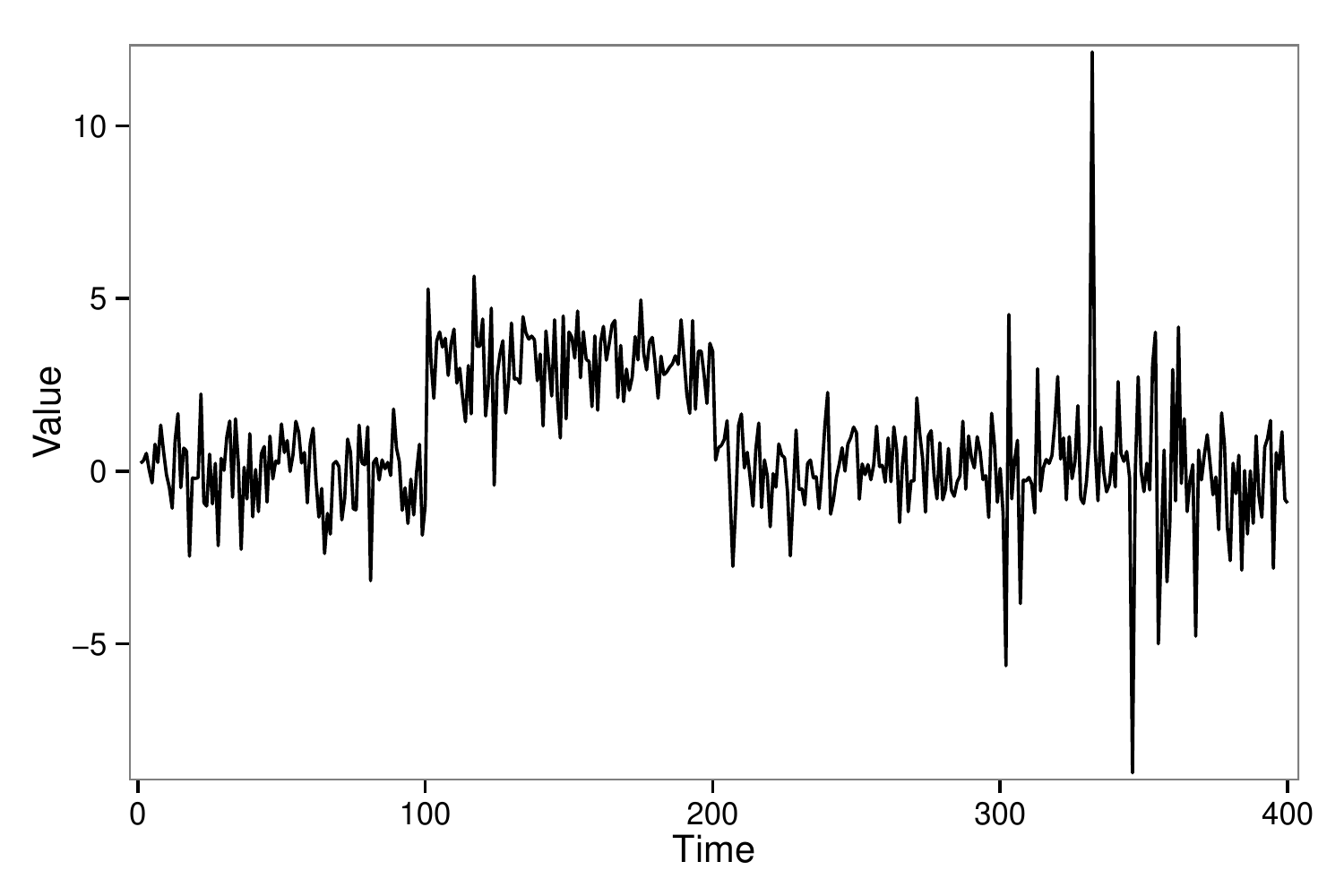}
		    \caption{\label{fig:sim1Pic} Example of time series with changes in mean and tail index. Mean changes occur at times 100 and 
		    200, while the tail index change is at time 300.}
		\end{figure}
	\end{minipage}
\end{minipage}
\par \ 

These two simulation studies on univariate time series show that the e-cp3o procedure performs well when compared to other parametric and nonparametric 
change point algorithms. The first set of simulations showed that it generated segmentations whose quality is comparable to that of an efficient 
parametric procedure when its parametric assumptions were satisfied. While the second set of simulations showed that it is able to handle more subtle 
distributional changes, such as a change in tail behavior. The flexibility of the e-cp3o method allows for it to be used when parametric assumptions are met, 
as well as in settings where they aren't sure to be satisfied.\par

\subsection{Multivariate Simulations}\label{subsec:MultiSim}
We now examine the performance of the e-cp3o procedure when applied to a multivariate time series. Since a change in mean can be seen as a change in a marginal distribution we 
could just apply any univariate method to each dimension of the dataset. For this reason we will examine a more complex type of distributional change. In this simulation 
the distributional change will be due to a change in the copula function \citep{sklar59}, while the marginal distributions remain unchanged. Since the PELT procedure as 
implemented in the \texttt{changepoint} package only performs marginal analysis it is not suited for this setting, and will thus not be part of our comparison. 
We instead consider a method proposed by \cite{gombay99} and implemented in the {\texttt\textbf{R}} package \texttt{npcp} by \cite{holmes13}. This package provides 
two methods that can be used in this setting. One that looks for any change in the joint distribution (NPCP-F) and one designed to detect changes in the copula function 
(NPCP-C).\par

For a given set of marginal distributions, the copula function is used to model their dependence. Thus a change in the copula function reflects a change in the dependence 
structure. This is of particular interest in finance where portfolios of dependent securities are typical \citep{guegan10}.\par

In this simulation we consider a two dimensional process where both marginal distributions are standard normal. While the marginal distributions remain static, the copula 
function evolves over time. For this simulation the copula undergoes two changes. Initially it is a Clayton copula and then changes to the independence copula and finally 
becomes a Gumbel copula. The density function for each of the used copulas is provided in Table \ref{table:copula} and simulation results in 
Table \ref{table:copulaSim}.

\begin{table}[ht]
\begin{center}
\begin{tabular}{lc}
\hline
Copula & Density $c(u,v)$\\
\hline
Clayton & $\left(\max\{u^{-2.8}+v^{-2.8}-1,\ 0\}\right)^{-5/14}$\\
Independence & $uv$\\
Gumbel & $\exp\left\{-\left[\left(-\log(u)\right)^{2.8} + \left(-\log(v)\right)^{2.8}\right]^{5/14}\right\}$
\end{tabular}
\end{center}
\caption{\label{table:copula}
The densities for the copula functions used in the multivariate simulations.}
\end{table}

\begin{figure}
	\centering
    \begin{subfigure}[b]{.3\textwidth}
		\includegraphics[width=\textwidth]{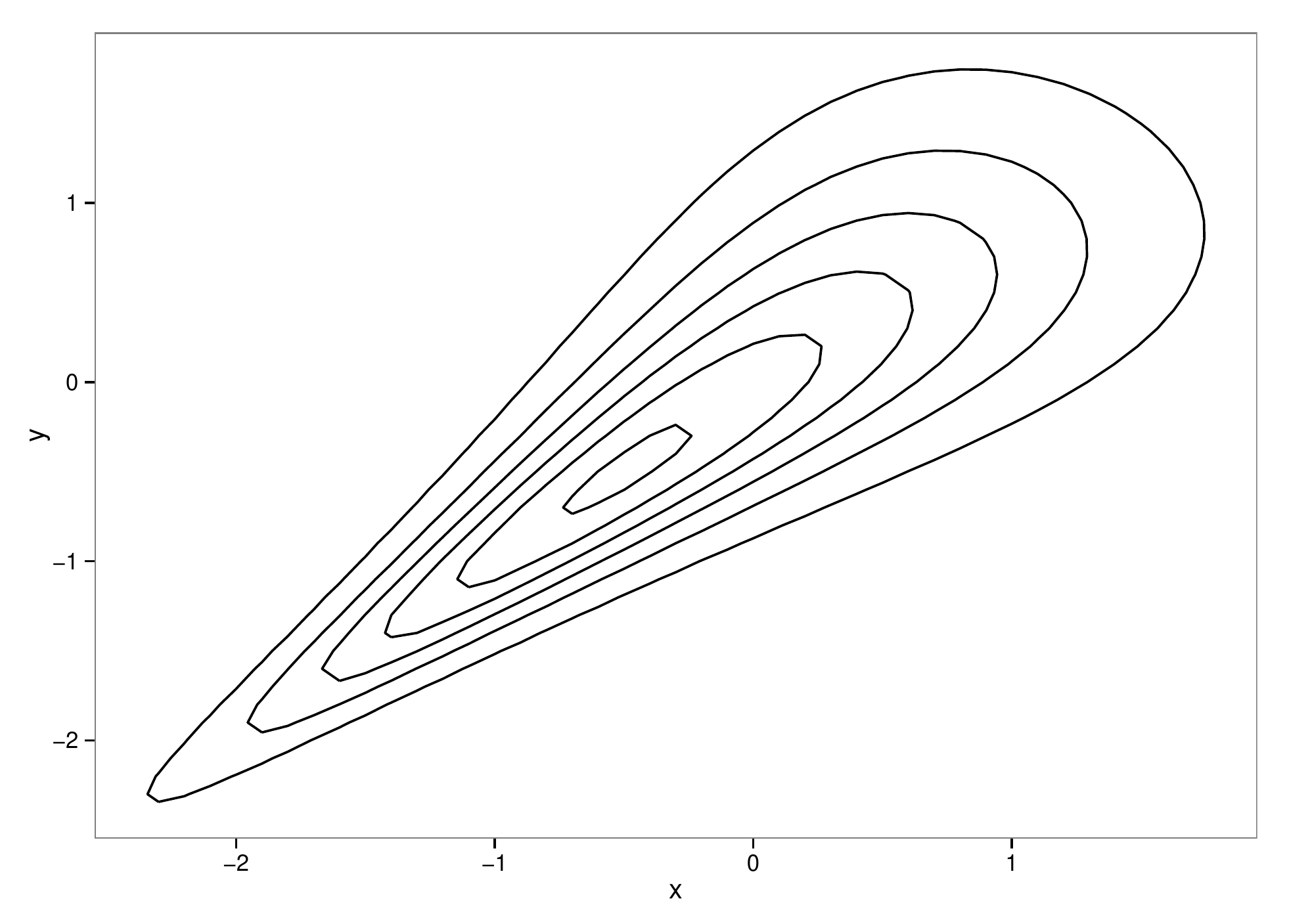}
		\caption{Clayton}
		\label{fig:claytonPic}
	\end{subfigure}
	\hspace{\fill}
	\begin{subfigure}[b]{.3\textwidth}
		\includegraphics[width=\textwidth]{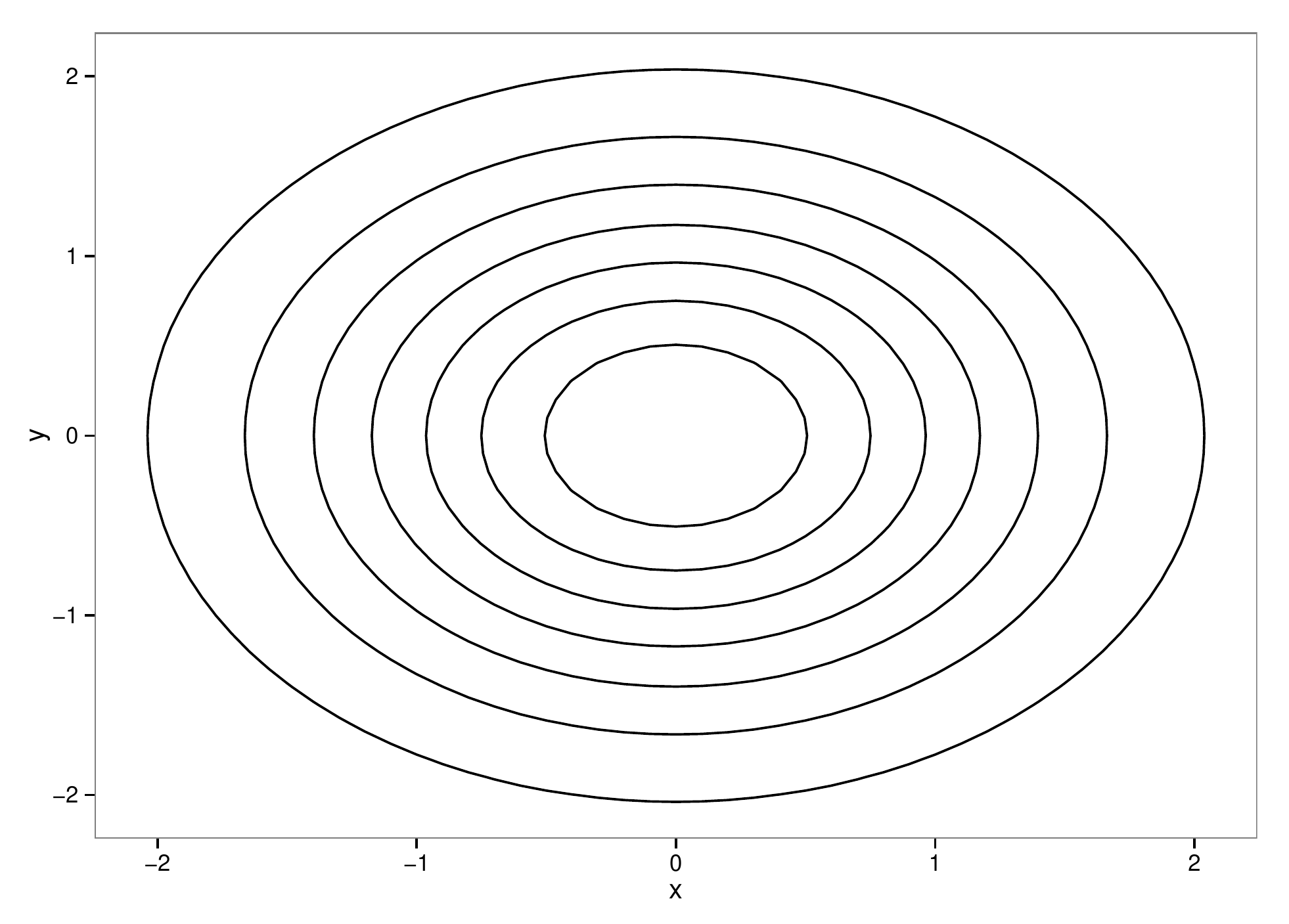}
		\caption{Independence}
		\label{fig:independencePic}
	\end{subfigure}
	\hspace{\fill}
	\begin{subfigure}[b]{.3\linewidth}
		\includegraphics[width=\linewidth]{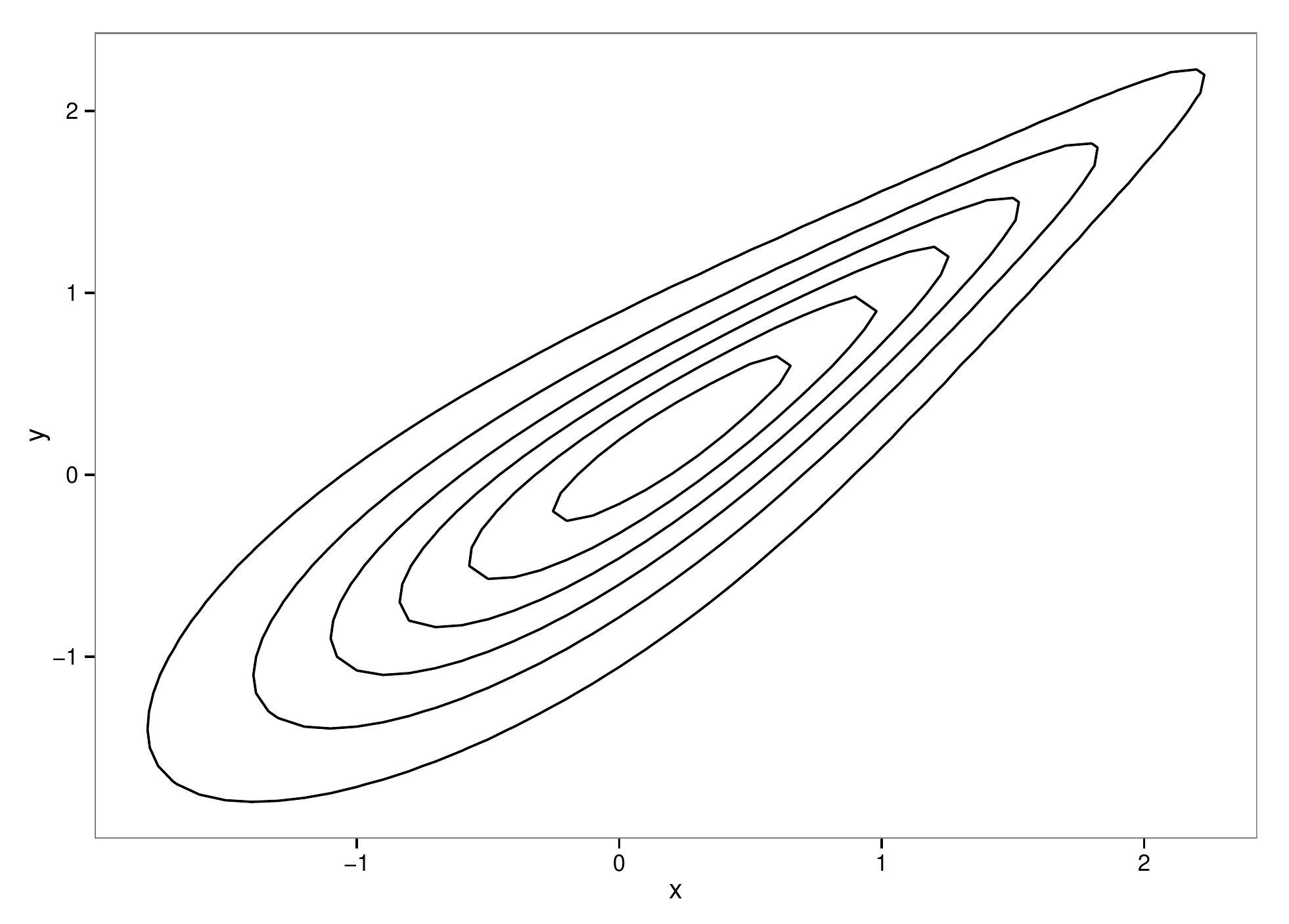}
		\caption{Gumbel}
		\label{fig:gumbelnPic}
	\end{subfigure}
	\caption{Contour plots of distributions used to generate data in the multivariate simulation.}
	\label{fig:contourPics}
\end{figure}

\begin{table}[H]
\centering
\begin{tabular}{lddd}
\hline
\multicolumn{1}{c}{} & \multicolumn{1}{c}{NPCP-C} & \multicolumn{1}{c}{NPCP-F} &\multicolumn{1}{c}{e-cp3o} \\ \hline
\multicolumn{1}{c}{} & \multicolumn{3}{c}{T=300, k(T)=2, K(T)=9}                                                                       \\ \cline{2-4}

Rand                                 &     0.958_{4\e{-3}}              &     0.616_{9\e{-3}}          &     0.685_{8\e{-3}}        \\
\# of cps                            &     2.130_{4\e{-2}}              &     0.320_{6\e{-2}}          &     4.000_{0.1}        \\
Time                                 &    73.163_{0.1}              &     1.871_{6\e{-2}}              &     0.125_{4\e{-3}}        \\ \hline

\multicolumn{1}{c}{} & \multicolumn{3}{c}{T=1200, k(T)=2, K(T)=9}                                                                       \\ \cline{2-4}

Rand                                 &0.979_{3\e{-3}}              &     0.865_{10^{-2}}      &   0.766_{10^{-2}}        \\ 
\# of cps                            &2.150_{4\e{-2}}             &     1.790_{9\e{-2}}       &  1.570_{7\e{-2}}        \\
Time                                 & 10,580.831_{4.0}              &    41.270_{0.7}             &   1.901_{5\e{-3}}        \\ \hline
\end{tabular}
\caption{\label{table:copulaSim}
Results of the multivariate simulation with different time series lengths. The subscripts indicate the standard errors for each value.
}
\end{table}

As was expected, in Table \ref{table:copulaSim} it is clear that the NPCP-C method obtained the best average Rand value in all situations. 
But this comes at a much increased average running time. This becomes very problematic when analysis of a single longer time series can take 
almost three hours. For shorter time series the e-cp3o provides the best combination between running time, estimated number of change points, 
and Rand value. For longer time series the NPCP-F procedure is the clear winner.\par


\section{Real Data}\label{Sec:Real}
In this section we apply the e-cp3o procedure to two real data sets. For our first application we make use of a dataset of monthly temperature 
anomalies. The second second consists of monthly foreign exchange (FX) rates between the United States, Russia, Brazil, 
and Switzerland.\par

\subsection{Temperature Anomalies}
For the first application of the e-cp3o procedure we examine the HadCRUT4 dataset of \cite{morice12}. This dataset consists of monthly global temperature anomalies from 1850 to 
2014. Since the dataset consists of anomalies, it does not indicate actual average monthly temperatures, but instead measured deviations from some predefined baseline. The time 
period used to create the baseline in this case spans 1960 to 1990.\par

The HadCRUT4 dataset contains two major components; one for land air temperature anomalies and another for sea-surface temperature anomalies. The analysis performed in this 
section will only consider the land air temperature anomaly component from the tropical region ($30^\circ$ South to $30^\circ$ North). This region was chosen because it was 
the most likely of all the presented regions to have a small difference between the minimum and maximum anomaly value, and be affected by changing seasons. More information about 
the dataset and the averaging process used can be found in the paper by \cite{morice12}.\par

From looking at the plot of the tropical land air anomaly time series it is suspected that there is some dependence between observations. This assumption is quickly confirmed by 
looking at the auto-correlation plot. As a result, we apply the e-cp3o procedure to the differenced data which visually appears to be piecewise stationary. The auto-correlation 
plot for the differenced data shows that much of the linear dependence has been removed, however, the same plot for the differences squared still indicates some dependence. As 
with the exchange rate data, we believe that this indicated dependence can be attributed to changes in distribution.\par

The e-cp3o procedure was applied with a minimum segment length of one year, corresponding to $\delta=11$; a maximum of $K(T)=20$ change points were fit, we chose $\alpha=1$, and 
$\epsilon=0.01$. Upon completion we identified change points at the following dates: July 1860, February 1878, January 1918, and February 1973, which are shown in Figure 
\ref{fig:climateTS}. With these change points we notice that the auto-correlation plots, for both the differenced and squared differenced data, show almost no statistically 
significant correlations. This is in line with our original hypothesis that the previously observed correlation was due to the distributional changes within the data.\par

\begin{figure}
    \centering
    \includegraphics[scale=.6]{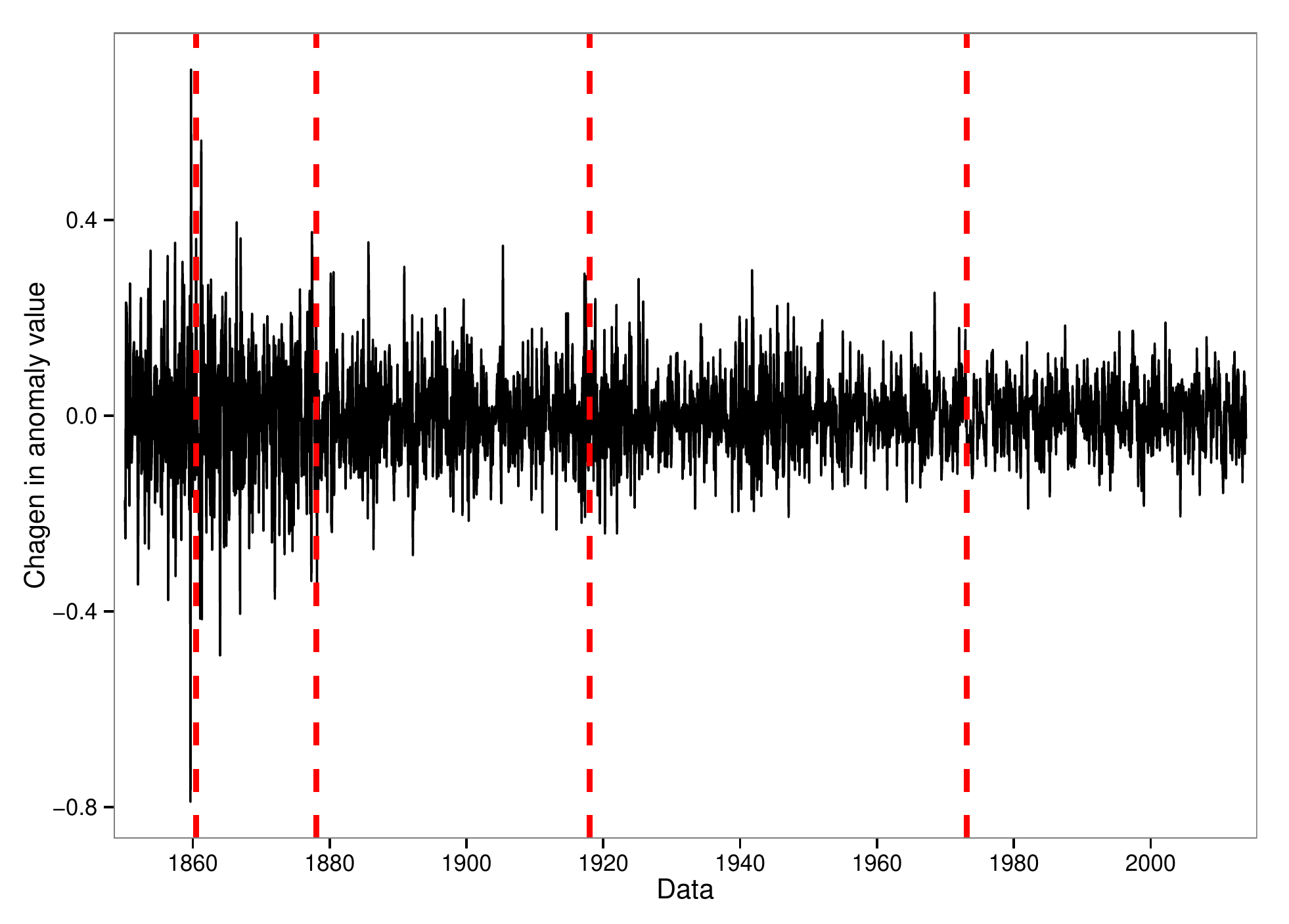}
    \caption{Change in land air temperature anomalies for the Tropical climate zone from February 1850 to December 2013. The cp3o estimated change point locations are 
    indicated by dashed vertical lines.}
    \label{fig:climateTS}
\end{figure}

Furthermore, the February 1973 change point occurs around the same time as the United Nations Conference on the Human Environment. This conference, which was held in June 1972, 
focused on human interactions with the environment. From this meeting came a few noteworthy agreed upon principles that have potential to impact land air temperatures:
\begin{enumerate}
    \item Pollution must not exceed the environment's ability to clean itself
    \item Governments would plan their own appropriate pollution policies
    \item International organizations should help to improve the environment
\end{enumerate}
These measures, undoubtedly played a role in the decreased average anomaly size, as well as an almost $66\%$ decrease in the variance.

\subsection{Exchange Rates}
We next apply the e-cp3o procedure to a set of spot FX rates obtained through the \texttt{R} package \texttt{Quandl} \citep{quandl13}. For our analysis we consider the three 
dimensional time series consisting of monthly FX rates for Brazil (BRL), Russia (RUB), and Switzerland (CHF). All of the rates are against the United States (USD). The time 
horizon spanned by this time series is September 30, 1996 to February 28, 2014, which results in a total of 210 observations. Looking at the marginal series it is obvious that 
each of the individual FX rates does not generate a stationary process. Thus, instead of looking at the actual rate, we look at the change in the log process. This transformation 
results in marginal processes that appear to at least be piecewise stationary.\par

Our procedure is only guaranteed to work with independent observations, so we must hope that our data either satisfies this condition or is very close to it. The papers by 
\cite{hsieh88,hsieh89} provide evidence that changes in the daily exchange rate are not independent, and that there is a reasonable amount of nonlinear dependence. However, they 
are not able to conclude whether this observed dependence is due to distributional changes or some other phenomena. For this reason we are instead interested in the change in the 
monthly exchange rate, which is more likely to either be weakly dependent or show no dependence. To check this we examine the auto/cross-correlation plots for both the difference 
and difference squared data. This preliminary analysis shows that there is no significant auto or cross-correlation within the differenced data, while for the squared differences 
there is only significant auto-correlation for Switzerland at a lag of one month.\par

The e-cp3o procedure is applied with a minimum segment length of six observations (half a year), which corresponds to a value of $\delta=5$. Furthermore, we have chosen to fit at 
most $K(T)=15$ change points, and values of $\alpha=1$ and $\epsilon=0.01$ were used. This specific choice of values resulted in change points being identified at May 31, 1998 
and March 31, 2000. These results are depicted in Figure \ref{fig:FXrates}.\par

It can be argued that changes in Russia's economic standing leading up to the 1998 ruble crisis are the causes of the May 31, 1998 change point. During the Asian financial crisis 
many investors were losing faith in the Russian ruble. At one point, the yield on government bonds was as high as $47\%$. This paired with a $10\%$ inflation rate would normally 
have been an investor's dream come true. However, people were skeptical of the government's ability to repay these bonds. Furthermore, at this time Russia was using a floating 
pegged rate for its currency, which resulted in the Central Bank's mass expenditure of USD's which further weakened the ruble's position.\par

The change point identified at March 31, 2000 also coincides with an economic shift in one of the examined countries. The country  most likely to be the cause of this change is 
Brazil. In 1994 the Brazilian government pegged their currency to the USD. This helped to stabilize the county's inflation rate; however, because of the Asian financial crisis 
and the ruble crisis many investors were averse to investing in Brazil. In January 1999 the Brazilian Central Bank announced that they would be changing to a free float exchange 
regime, thus their currency was no longer pegged to the USD. This change devalued the currency and helped to slow the current economic downturn. The change in exchange regime and 
other factors led to a $48\%$ debt to GDP ratio, besting the IMF target and thus increasing investor faith in Brazil.

\begin{figure}
    \centering
    \begin{subfigure}[b]{.3\linewidth}
        \includegraphics[width=\linewidth]{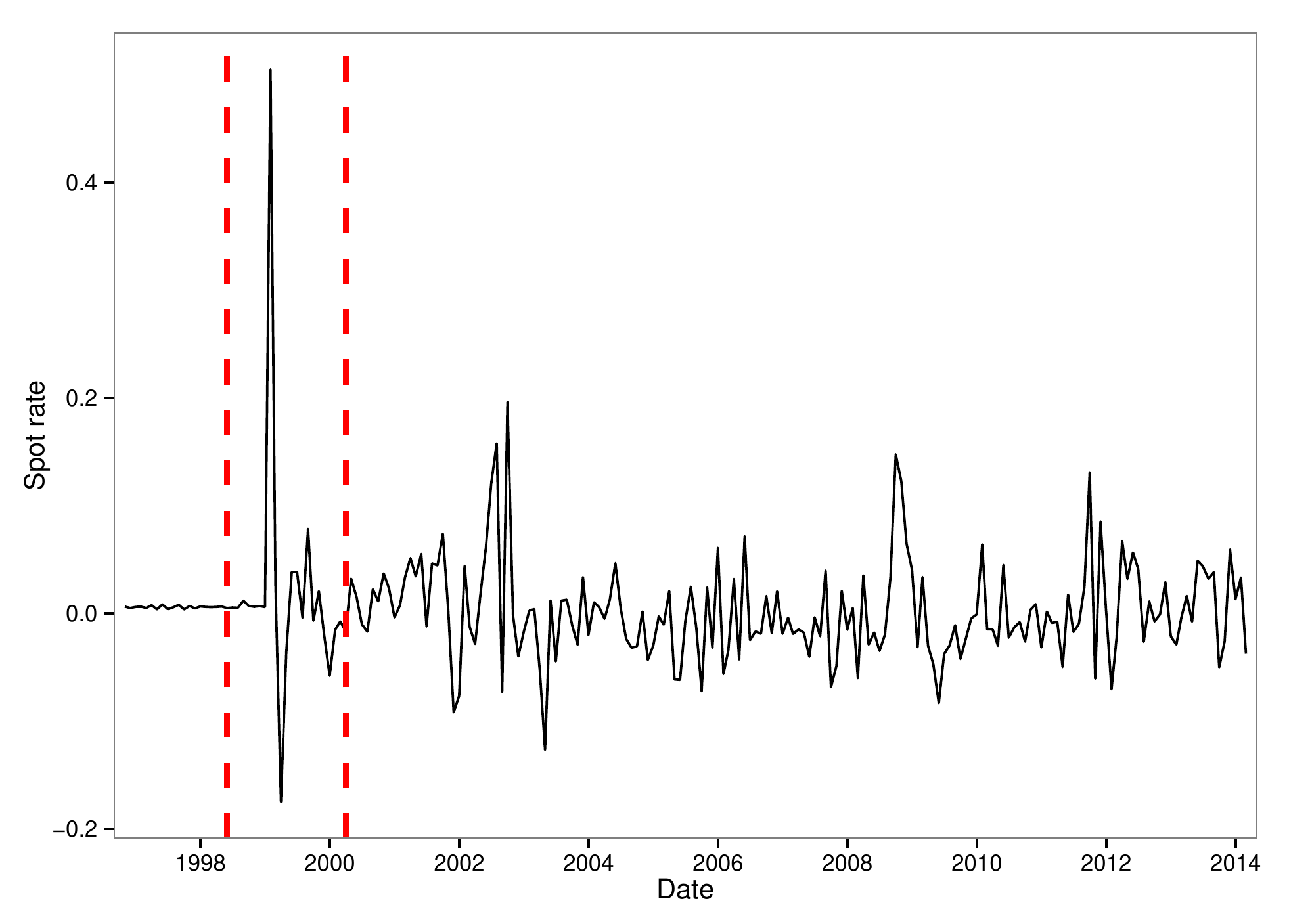}
        \caption{Brazil}
        \label{fig:braFX}
    \end{subfigure}
    \hspace{\fill}
    \begin{subfigure}[b]{.3\linewidth}
        \includegraphics[width=\linewidth]{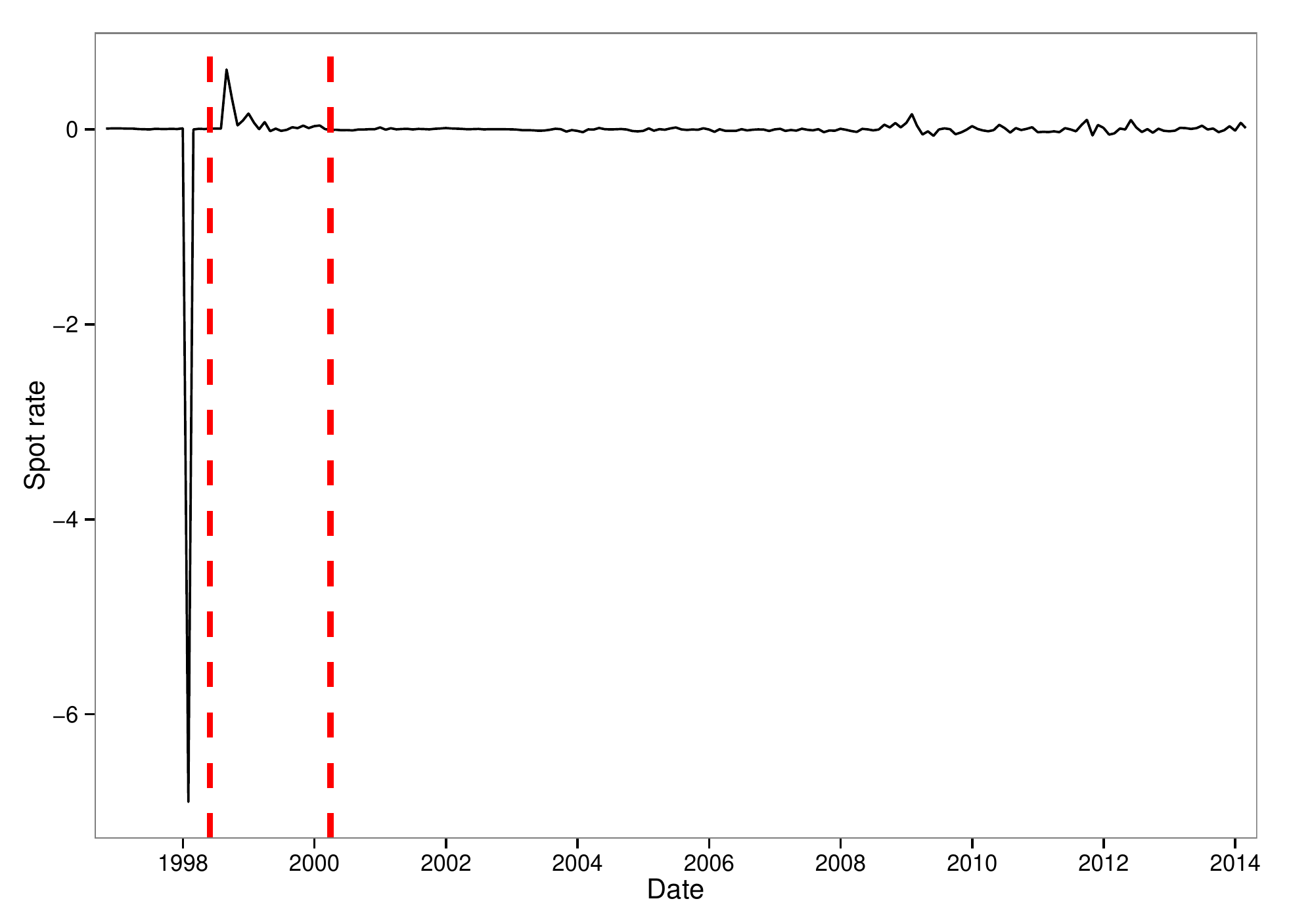}
        \caption{Switzerland}
        \label{fix:swsFX}
    \end{subfigure}
    \hspace{\fill}
	\begin{subfigure}[b]{.3\linewidth}
       \includegraphics[width=\linewidth]{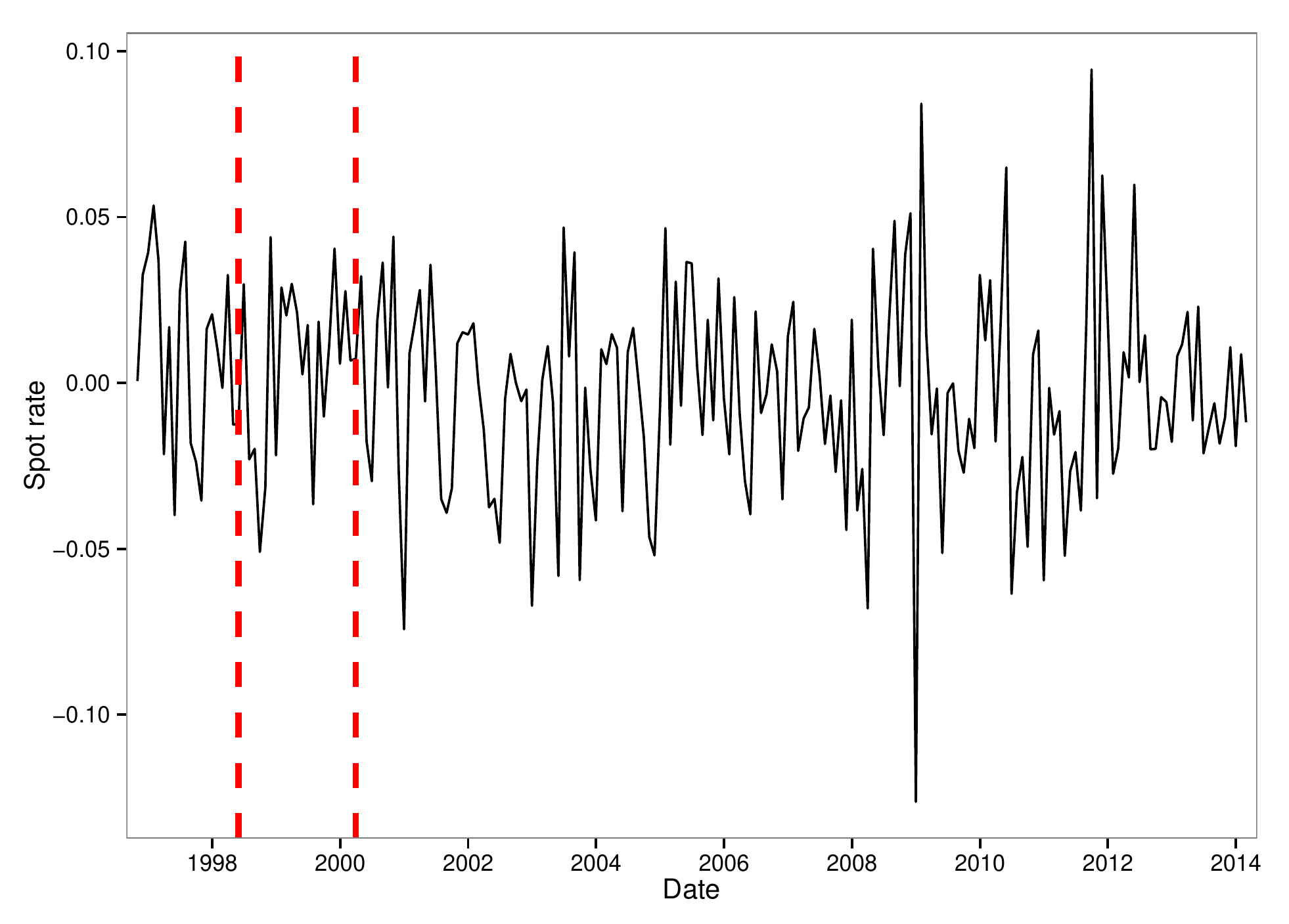}
       \caption{Russia}
       \label{fig:rusFX}
    \end{subfigure}
    \caption{Time series for FX spot rates for each of the three countries. Estimated change point locations indicated by vertical lines.}
    \label{fig:FXrates}
\end{figure}

\section{Conclusion}\label{Sec:Conclusion}
We have presented an exact search algorithm that incorporates probabilistic pruning in order to reduce the amount of unnecessary calculations. This search method 
can be used with almost any goodness of fit measure in order to identify change points in multivariate time series. Asymptotic theory has also been provided showing 
that the cp3o algorithm can generate consistent estimates for both the number of change points as well as the change point locations as the time series increases, 
provided that a suitable goodness of fit measure is provided. Furthermore, the decoupling of the search procedure and the determination of the number of estimated change 
points allows for the cp3o algorithm to efficiently generate a collection of optimal segmentations, with differing numbers of change points. This is all accomplished 
without the user having to specify any sort of penalty constant or function.\par

By combining the cp3o search algorithm with E-Statistics we developed e-cp3o, a method to perform nonparametric multiple change point analysis that can detect 
\textit{any} type of distributional change. This method combines an approximate statistic with an exact search algorithm. The slight loss in accurately estimating 
change point locations on finite time series is greatly outweighed by the dramatic increase in speed, when compared to similar methods that combine an exact statistic 
with an approximate search algorithm.\par

\bibliographystyle{JASA}
\bibliography{fast}
\end{document}